\documentclass[11pt]{article}

\usepackage[margin=1in]{geometry}
\usepackage[colorlinks]{hyperref}
\usepackage[ruled,vlined,linesnumbered]{algorithm2e}
\usepackage{amsmath, amsthm, amssymb, latexsym}
\usepackage{color}
\usepackage{graphicx}
\usepackage{thm-restate}
\usepackage{subcaption}
\usepackage[capitalize,nameinlink,noabbrev]{cleveref}

\usepackage{tikz}
\usetikzlibrary{shapes}
\usepackage{enumitem}

\usepackage{xcolor}         % For coloring
\usepackage{etoolbox}       % For patching commands

\hypersetup{
    colorlinks=true,
    linkcolor=blue!40!black,     % For \ref and \eqref
    citecolor=teal!60!black,     % For \cite
    urlcolor=cyan!50!black       % For \url
}

\newcommand{\keywords}[1]{\vspace{1em}\noindent\textbf{Keywords:} #1}

\newcommand{\remark}[1]{\par\noindent\textbf{Remark.} #1\par}

\newtheorem{theorem}{Theorem}[section]
\newtheorem{lemma}[theorem]{Lemma}
\newtheorem{claim}[theorem]{Claim}

\newtheorem{proposition}[theorem]{Proposition}

\theoremstyle{definition}
\newtheorem{definition}[theorem]{Definition}
\theoremstyle{plain}

\newcommand{\Exp}{\mathop{\mathbb{E}}}

\newcommand{\OPT}{\mathrm{OPT}}
\newcommand{\COPT}{\mathrm{COPT}}
\newcommand{\TOPT}{\mathrm{TOPT}}
\newcommand{\ALG}{\mathrm{ALG}}
\newcommand{\SOL}{\mathrm{SOL}}
\newcommand{\proj}{\mathtt{proj}}

\title{Online Connectivity Augmentation}
\author{Mohit Garg\thanks{Indian Institute of Science, Bengaluru, India ({garg@uni-bremen.de}). This work was initiated when the first author was supported by a fellowship from the Walmart Center for Tech Excellence at IISc (CSR Grant WMGT-23-0001).}
    \and Aditya Subramanian\thanks{Indian Institute of Science, Bengaluru, India ({adityasubram@iisc.ac.in}). Supported by the Walmart Center for Tech Excellence at IISc (CSR Grant WMGT-23-0001).}}
\date{}

\begin{document}
\maketitle
\thispagestyle{empty}

\begin{abstract}
The {\sc Connectivity Augmentation Problem} (CAP) is a fundamental problem in fault-tolerant network
design and has been extensively studied in the context of approximation algorithms. In this work, we
consider CAP in the \emph{online setting}: given a $k$-edge-connected graph $G$ with $n$ vertices
and a set $L$ of additional edges over the vertices of $G$, called links, online requests arrive one
by one, each specifying two vertices that need to be $(k+1)$-edge-connected. We start with the graph
$G$ and progressively add links to serve these requests. More specifically, upon the arrival of a
request $\{u,v\}$, we must \emph{immediately and irrevocably} add zero or more links from $L$ to the
graph so that $u$ and $v$ are $(k+1)$-edge-connected in the resulting augmented graph. The goal is
to minimize the total number of links added, and we evaluate an algorithm’s performance by its
\emph{competitive ratio} relative to an optimal offline solution.

Prior works by Gupta, Krishnaswamy, and Ravi (2009) and Naor, Umboh, and Williamson (2019) imply the
following bounds on the competitive ratio for online CAP: a randomized $\tilde{O}(k\log^3 n)$ upper
bound, a deterministic $O(\log n)$ upper bound for the special case $k = 1$ (also known as the
online {\sc Tree Augmentation Problem} (TAP)), and an $\Omega(\log n)$ lower bound. These bounds
also extend to the \emph{weighted setting}, where links have weights and the objective is to
minimize the total weight of the links added. We show the following.

\begin{description}
    \item[For Online CAP:] We present a deterministic $O(\log n)$-competitive algorithm, improving
        over the previously known randomized $\tilde{O}(k \log^3 n)$ bound. This result is tight up
        to constant factors, as an $\Omega(\log n)$ lower bound was previously known.
    \item[For Online CAP in the Random-Order Model:] In this model, the adversary fixes the request
        sequence in advance, and the requests are revealed to the algorithm in a uniformly random
        order. This restriction weakens the adversary, potentially allowing improved competitive
        ratios. However, we show that this is not possible: we prove an $\Omega(\log n)$ lower bound
        on the competitive ratio in this setting. Combined with our upper bound, this is tight up to
        constant factors. Notably, our lower bound holds even for TAP, the simplest possible
        setting.
    \item[For Online Weighted CAP:] We present a deterministic $O(\log^2 n)$-competitive algorithm.
        This improves upon the previously known randomized $\tilde{O}(k\log^3 n)$ bound.
    \item[For Online TAP:] We provide a deterministic $O(\log n)$-competitive algorithm with a
        simple combinatorial analysis, distinct from the one known earlier.
\end{description}
\end{abstract}


\newpage
\setcounter{page}{1}

\section{Introduction}

Networks deployed in the real world are prone to failures. Fault-tolerant network design aims to
construct networks that maintain a desired level of connectivity between nodes or groups of nodes,
while being resilient to a certain number of node or edge failures. In a graph $G$, we say that two
vertices $u$ and $v$ are $k$-edge-connected if deleting any set of fewer than $k$ edges does not
disconnect them. Equivalently, by Menger's theorem, $u$ and $v$ are $k$-edge-connected if and only
if there exist at least $k$ edge-disjoint paths between them. Furthermore, a graph is said to be
$k$-edge-connected if every pair of distinct vertices in the graph is $k$-edge-connected. One of the
most fundamental and extensively studied problems in fault-tolerant network design is the {\sc
Connectivity Augmentation Problem} (CAP). 
In CAP, the input consists of an undirected $k$-edge-connected graph $G$ with $n$ vertices and a set
of additional edges (called links) $L \subseteq \binom{V(G)}{2}$. The goal is to find a subset $L'
\subseteq L$ of minimum cardinality such that every pair of distinct vertices is
$(k+1)$-edge-connected in the augmented graph $G \cup L'$; i.e., $G \cup L'$ is
$(k+1)$-edge-connected.

In this work, we initiate the study of CAP in the \emph{online setting}, where we are given an
undirected $k$-edge-connected graph $G$ along with a set of links $L$, and online requests arrive
one by one, each specifying two vertices that need to be $(k+1)$-edge-connected. We start with the
graph $G$ and progressively add links to serve these requests. More specifically, upon the arrival
of a request $\{u,v\}$, we must \emph{immediately and irrevocably} add zero or more links from $L$
to $G$ so that $u$ and $v$ are $(k+1)$-edge-connected in the resulting augmented graph. The goal is
to minimize the total number of links added, and we evaluate an algorithm’s performance by its
\emph{competitive ratio} relative to an optimal offline solution.

In contrast to the unweighted setting, the weighted version, denoted WCAP, assigns non-negative
weights to links, and the objective is instead to minimize the total weight of the links added.

CAP, and consequently its weighted variant WCAP, are APX-hard even for $k = 1$, and have been
extensively studied in the context of approximation algorithms. The special case of CAP when $k = 1$
is known as the {\sc Tree Augmentation Problem} (TAP), since in this case, any cycle in the input
graph can be safely contracted to a single vertex, as cycles are already $2$-edge-connected;
repeatedly contracting such cycles reduces the input graph to a tree. This case has also been
extensively studied. The corresponding weighted variant is referred to as {\sc Weighted TAP} (WTAP).
For all four problems---CAP, WCAP, TAP, and WTAP---multiple 2-approximations (see, for
example,~\cite{J01}) were long known, and improving this ratio was considered a significant
challenge. In a long line of research, progressively better-than-2 approximation algorithms have
been developed~\cite{A19,CG18a,CG18,EFKN09,FGKS18,GKZ18,KN16,KN16b,N03,BGJ20,CTZ21,TZ21}. The
current best-known approximation ratio is 1.393 for both CAP and TAP~\cite{CTZ21}, and $1.5 +
\varepsilon$ for WCAP~\cite{TZ23} and WTAP~\cite{TZ22}, for any constant $\varepsilon > 0$.

The online version of connectivity augmentation has previously been studied only for the special
case $k = 1$, where Naor, Umboh, and Williamson~\cite{NaorUW22} presented a deterministic $O(\log
n)$-competitive algorithm for online WTAP. This guarantee is tight up to constant factors. For
arbitrary values of $k$, bounds can be inferred from a more general problem studied by Gupta,
Krishnaswamy, and Ravi~\cite{GuptaKR12}: the online {\sc Survivable Network Design Problem} (SNDP).
In their setting, the underlying graph is fixed with non-negative edge weights, and online requests
arrive in the form of vertex pairs $\{s_i, t_i\}$, each with an associated edge-connectivity
requirement $r_i$. The task is to maintain a subgraph that satisfies the required number of
edge-disjoint paths for all requests seen so far. They provide a randomized $\tilde{O}(r_{\max}
\log^3 n)$-competitive algorithm, where $r_{\max} = \max_i r_i$. Observe that WCAP can be embedded
as a special case by treating the given $k$-edge-connected graph as a fixed subgraph with
zero-weight edges and the links as additional weighted edges, with all edge-connectivity
requirements equal to $k+1$. To avoid parallel edges that could arise from superimposing the fixed
graph and the links, we can subdivide each edge of the original graph into two zero-weight edges.
Under this embedding, their result yields a $\tilde{O}(k \log^3 n)$-competitive algorithm for online
WCAP.

\subsection{Our Results}

We obtain tight or near-tight bounds for several online variants of the connectivity augmentation
problem, including its weighted version and the random-order model. Our main results are summarized
below.

\begin{itemize}

  \item \textbf{Online CAP.}  
  We present a deterministic $O(\log n)$-competitive algorithm for online CAP. This improves upon
  the previously known randomized $\tilde{O}(k \log^3 n)$ bound and matches the known $\Omega(\log
  n)$ lower bound up to constant factors.

  \begin{restatable}{theorem}{oncap}
  \label{thm:online-cap}
  There exists a deterministic $O(\log n)$-competitive algorithm for online CAP.
  \end{restatable}

  \item \textbf{Online CAP in the Random-Order Model.}  
  In this model, the adversary fixes the entire request sequence in advance, and the algorithm
  receives requests in a uniformly random order. One might hope this weaker adversary allows better
  competitive ratios. However, we show that the $\Omega(\log n)$ hardness persists even against this
  weaker adversary.

  \begin{restatable}{theorem}{randorderlower}
  \label{thm:random-order-lower}
  Any algorithm for online CAP in the random-order model has competitive ratio $\Omega(\log n)$.
  This lower bound holds even for the special case of online TAP.
  \end{restatable}

  \item \textbf{Online WCAP.}  
  We give a deterministic $O(\log^2 n)$-competitive algorithm for online WCAP, improving on the
  prior randomized $\tilde{O}(k \log^3 n)$ bound.

  \begin{restatable}{theorem}{wcaponline}
  \label{thm:weighted-cap}
  There exists a deterministic $O(\log^2 n)$-competitive algorithm for online WCAP.
  \end{restatable}

  \item \textbf{Online TAP.}  
  We also provide a new $O(\log n)$-competitive algorithm for online TAP, featuring a simple and
  self-contained combinatorial analysis.

  \begin{restatable}{theorem}{taponline}
  \label{thm:online-tap}
  There exists a deterministic $O(\log n)$-competitive algorithm for online TAP.
  \end{restatable}

\end{itemize}

\remark{The previous competitive ratio for WCAP, $\tilde{O}(k \log^3 n)$, scaled linearly with $k$,
which could itself grow with $n$. Our results eliminate this dependence on $k$ entirely. Moreover,
our online algorithms run in polynomial time, unlike the earlier method for WCAP, whose running time
was not guaranteed to be polynomial when $k = \omega(1)$.}

\subsection{Our Technique}

We obtain our upper bounds for online CAP and WCAP by first leveraging a classical result of Dinitz,
Karzanov, and Lomonosov~\cite{DKL76}, which reduces any $k$-edge-connected graph to a cactus while
preserving all global minimum cuts. This enables a reduction to the {\sc Cactus Augmentation
Problem} (CacAP). CacAP is a special case of CAP with $k = 2$ where the input graph is a
cactus---that is, a connected graph in which every edge lies on exactly one cycle. The reduction is
known to be approximation-preserving in the offline setting; we show that it can be adapted to the
online setting while preserving competitiveness.

To solve online CacAP and its weighted version (WCacAP), we draw inspiration from an offline
reduction of CAP to the Steiner tree problem~\cite{BasavarajuFGMRS14,BGJ20}. Our reduction
follows a similar high-level approach but introduces new technical elements tailored to the online
setting. Specifically, since connectivity must be augmented only for online-requested pairs (and not
all vertex pairs), we reduce to the online {\sc node-weighted Steiner forest} (NWSF) problem.
Crucially, we include not only degree-2 vertices (as in the offline reduction) but also all
articulation vertices---i.e., vertices that participate in multiple cycles of the cactus---as
Steiner terminals. This refinement is essential for preserving feasibility and competitiveness.

By applying the known $O(\log^2 n)$-competitive algorithm for online NWSF~\cite{Borst0V25}, we
obtain the same bound for online WCacAP, and consequently for online WCAP.

To obtain our $O(\log n)$-competitive algorithm for online CAP, we begin with the special case where
the cactus consists of a single cycle, referred to as the online {\sc Cycle Augmentation Problem}
(CycAP). Here, we reduce to the online {\sc edge-weighted Steiner forest} (EWSF) problem, which
admits a deterministic $O(\log n)$-competitive algorithm~\cite{BermanC97}. The reduction is
facilitated by the fact that all Steiner nodes (i.e., links) have degree 2 and unit weight, allowing
us to simulate node-weighted selection through edge-weighted selection with only a constant-factor
loss in competitiveness.

Our key technical contribution lies in handling the general CacAP instance. Since a cactus graph can
be viewed as a tree of cycles joined at articulation vertices, we reduce the full problem to (i) one
instance of online TAP over a tree defined by these articulation vertices, and (ii) one instance of
CycAP for each cycle in the cactus. These instances are solved independently using known algorithms
for online TAP and online Steiner forest. The idea is that the algorithm handling the cycle instances performs local optimizations within individual cycles, whereas the algorithm for the tree instance manages global, long-range optimizations across the cactus structure. We carefully analyze how each request decomposes into
subrequests handled by the tree and the cycles, and show via a charging argument that the combined
cost is within $O(\log n)$ of the offline optimum. This decomposition and analysis are the central
innovations of our work and yield a tight bound for online CAP.

To complement our upper bounds, we establish an $\Omega(\log n)$ lower bound on the competitive
ratio of online TAP, even under random-order arrivals, showing that the inherent hardness persists
despite the weaker adversary. The construction is based on a complete binary tree in which each edge
is replaced by a path whose length decreases geometrically with its level, resulting in a slightly
larger tree. Links are placed between the root and each leaf, and the online requests consist of
pairs of consecutive vertices along a uniformly random root-to-leaf path. As a result, the offline
optimum can satisfy all requests using a single link. A key insight is that requests from only the higher
levels are likely to appear early in the sequence, revealing little about the specific path being
followed. This uncertainty forces the algorithm to commit to incorrect links, incurring an expected
cost of $\Omega(\log n)$---matching our upper bounds up to constant factors.

Finally, we present a simple $O(\log n)$-competitive algorithm for online TAP, along with a
self-contained analysis. In contrast, the earlier algorithm by Naor et al.~\cite{NaorUW22}, which
also applies to the weighted version (online WTAP), is based on a more involved primal-dual
framework. For online TAP (and WTAP), each request can be assumed---without loss of generality---to
correspond to an edge in the underlying tree. For each such request, the algorithm selects two
links: one that connects to the highest ancestor and another that descends to the deepest reachable
vertex. The ``deepest'' vertex is determined via a greedy traversal, always descending into the
subtree with the most nodes. The key intuition is that each non-optimal link selected by the
algorithm eliminates a substantial portion of the tree, thereby limiting the adversary's ability to
mislead the algorithm in future steps. This structural property allows us to bound the total number
of such mistakes, ultimately yielding the desired competitive ratio.

\subsection{Further Related Work}
Connectivity and Tree Augmentation problems are closely related to several fundamental network
design problems. In the {\sc 2-Edge-Connected Spanning Subgraph} (2-ECSS) problem, the goal is to
find a minimum-cardinality spanning subgraph that is 2-edge-connected. Numerous better-than-2
approximation algorithms are known~\cite{KV94,CSS01,HVV19,SV14,GGJ23soda,KN23,GVS93,GargHL24}, for
this problem, with the current best approximation factor being
$\frac{5}{4}$~\cite{Bosch-Calvo00HA25}. However, obtaining a better-than-2 approximation for the
weighted version of 2-ECSS remains a major open problem. WTAP can be viewed as special
cases of the weighted 2-ECSS problem.

Other well-studied problems that generalize 2-ECSS and fall under the umbrella of weighted 2-ECSS
include the {\sc Matching Augmentation Problem} (MAP), studied
in~\cite{BDS22,CDGKN20,CCDZ23,GHM23}), and the {\sc Forest Augmentation Problem} (FAP), studied
in~\cite{GJT22,Felix25}. In MAP, the input includes a graph with a matching as a weight-zero
subgraph; for FAP, the zero-weight subgraph is a forest. The objective in both cases is to augment
the input with a minimum-cost set of additional edges to achieve 2-edge connectivity. These
augmentation problems and their weighted variants also present natural and compelling directions for
study in the online setting.

Connectivity in the online setting has also been extensively studied, especially in the context of
of Steiner tree and Steiner forest and their variants. See, for example,
\cite{HajiaghayiLP17,Borst0V25} and the references therein.

\subsection*{Organization of the Paper}
We begin by introducing the basic terminology and stating a few essential propositions in \Cref{sec:prelim}. 

In \Cref{sec:onlineCAP}, we present the $O(\log n)$-competitive algorithm for online CAP (\Cref{thm:online-cap}), assuming certain statements about the Cycle (\Cref{thm:unwtcycap}) and Tree Augmentation (\Cref{thm:online-tap}) problems, which are proved later. 

\Cref{sec:ROlowerbound} establishes an $\Omega(\log n)$ lower bound for CAP under random-order arrivals (\Cref{thm:random-order-lower}). 

In \Cref{sec:logsqnWCAP}, we first give a reduction from CacAP to the node-weighted Steiner forest problem, thereby obtaining an $O(\log^2 n)$-competitive algorithm for WCAP (\Cref{thm:logsqcacap}). We then apply a black-box reduction to the edge-weighted Steiner forest problem to derive an $O(\log n)$-competitive algorithm for CycAP (\Cref{thm:unwtcycap}). 

Finally, in \Cref{sec:onlineTap}, we present a simple $O(\log n)$-competitive algorithm for TAP (\Cref{thm:online-tap}). 

The proofs of the propositions stated in the preliminaries are provided in \Cref{app:captocacap}, \Cref{sec:transitivity}, and \Cref{app:connectedCrossing}.

%%%%%%%%%%%%End of Section%%%%%%%%%

\section{Preliminaries} \label{sec:prelim}

The performance of an online algorithm is measured by its \emph{competitive ratio}.

\begin{definition}[Competitive Ratio]
Let $\ALG$ be an online algorithm for connectivity augmentation.  We say that $\ALG$ is
\emph{$c$-competitive} if, for every $k$-edge-connected input graph $G$, link set $L$, and request
sequence $\sigma$,
\[
\ALG(G,L,\sigma)\; \le\; c \,\cdot\, \OPT(G,L,\sigma),
\]
where $\ALG(G,L,\sigma)$ is the cost incurred by $\ALG$ and $\OPT(G,L,\sigma)$ is the cost of an
optimal offline solution on the same instance.
\end{definition}

For randomized algorithms, or when the adversary reveals the requests in a uniformly random order, the
inequality is understood in expectation, i.e., the left-hand side is $\Exp[\ALG(G,L,\sigma)]$.

We now formally define a cactus which is central to our work and analysis.

\begin{definition}[Cactus Graph and Articulation Vertex]
A \textbf{cactus} is a connected graph in which every edge lies in exactly one cycle. A vertex that
lies on two or more cycles of the cactus is called an \textbf{articulation vertex}.
\end{definition}

\noindent\textbf{Reduction to cactus augmentation.}
Cactus graphs play a central role in our approach because the following reduction transforms any
connectivity augmentation instance into an equivalent cactus augmentation instance, while preserving
the competitive ratio.

\begin{restatable}{proposition}{CapFromCacap}
\label{thm:cap-from-cacap}
If there exists a $c$-competitive online algorithm for the cactus augmentation problem (CacAP or
WCacAP), then there exists a $c$-competitive online algorithm for the connectivity augmentation
problem (CAP or WCAP, respectively). 
\end{restatable}

Such a reduction is known in the offline setting via~\cite{DKL76}, and we adapt it to the online
setting. The proof of \Cref{thm:cap-from-cacap} is provided in \Cref{app:captocacap}.

\noindent\textbf{Cactus augmentation machinery.} Given that connectivity augmentation reduces to
cactus augmentation, we discuss some basic definitions and propositions that will help us in
developing and analyzing online algorithms for cactus augmentation. These definitions and
propositions are adapted from~\cite{BGJ20}, with minor modifications where necessary.

Given a graph $G$, a cut is a partition of $V(G)$ into two parts. 
Given a cut $(V_1, V_2)$ in a cactus graph, we say that a link 
$\ell = \{u, v\}$ \textit{covers} the cut if $u \in V_1$ and $v \in V_2$, or vice versa.
A key component of our approach for cactus graphs
is to decompose any given link based on the cycles a path in the cactus between the link's endpoints traverses. This decomposition gives rise to
the notions of a link's projection vertex set and its projections onto cycles. Sometimes it will be
useful to decompose a pair of distinct vertices which need not be a link, so the following
definition is stated more generally.
%,and its head, tail, and middle segments.

\begin{definition}[Link Decomposition]
Let $G$ be a cactus and let $\ell=(u,v)$ be an arbitrary vertex pair consisting of two distinct
vertices. Note that although links are inherently unordered pairs, we will, for technical
convenience, sometimes treat them as ordered.

\begin{enumerate}
    \item \textbf{Projection Vertex Set:}
        Consider all simple paths between $u$ and $v$ in the cactus graph.
        Let $(c_1, c_2, \dots, c_k)$ be the sequence of articulation vertices (excluding $u$ and
        $v$) that appear on every such path from $u$ to $v$, in order. The \textbf{Projection
        Vertex Set} of $\ell$ is the ordered set of vertices $\mathcal{V}(\ell) = (u, c_1, \dots,
        c_k, v)$.

    \item \textbf{Projection on a Cycle:} The Projection Vertex Set $\mathcal{V}(\ell)$ decomposes
        $\ell=(u,v)$ %the path between $u$ and $v$
        into a set of segments. Each segment $(x,y)$, where $x, y$ are
        consecutive vertices in $\mathcal{V}(\ell)$, lies on a unique cycle $C$. We call such a
        segment $(x,y)$ the \textbf{projection} of $\ell$ onto the cycle $C$, denoted
        $\proj_C(\ell)$.
\end{enumerate}
\end{definition}

The projection vertex set of a link $\mathcal{V}(\ell)$ characterizes which vertex pairs become
$3$-edge-connected when $\ell$ is added to $G$.  
We have the following two simple propositions.

\begin{restatable}{proposition}{transitivity}\label{prop:transitivity}
    If vertices $u$ and $v$ become $3$-edge-connected in the cactus $G$ after adding some links,
    then every pair of distinct vertices $(x,y) \in \mathcal{V}((u,v)) \times \mathcal{V}((u,v))$ is
    also $3$-edge-connected.
\end{restatable}

The above proposition is proved in \Cref{sec:transitivity}.
Moreover, when we augment $G$ with a single link $\ell$, its endpoints immediately become
$3$-edge-connected, leading to the following immediate consequence.

\begin{proposition}
    If a link $\ell$ is added to a cactus $G$, then every pair of distinct vertices $(u,v) \in
    \mathcal{V}(\ell) \times \mathcal{V}(\ell)$ has three edge-disjoint paths between them.
\end{proposition}

We now define the notion of when two links ``cross".

\begin{definition}[Crossing Links] \label{def:links}
Let $G=(V,E)$ be a cactus graph.
\begin{enumerate}
    \item \textbf{Within a cycle:} Let $\ell = (x, y)$ and $\ell' = (x', y')$ be two links whose
        endpoints are all in the same cycle $C$. We say that $\ell$ and $\ell'$ \textbf{cross in C}
        if one of the following two conditions holds:
    \begin{enumerate}
        \item The links share exactly one endpoint (e.g., $y=x'$ and $x \neq y'$).
        \item The endpoints of the two links are distinct and they interleave on the cycle. That is,
            traversing a simple path from $x$ to $y$ along an arc of $C$ encounters exactly one
            vertex from $\{x', y'\}$ as an internal vertex.
    \end{enumerate}
    \item \textbf{General case:} Two links $\ell$ and $\ell'$ are said to \textbf{cross} if they
        share exactly one common end-point, or if there exists a projection $\proj_C(\ell)$ of
        $\ell$ and a projection $\proj_C(\ell')$ of $\ell'$ onto the same cycle $C$, such that these
        two projections cross in $C$.
\end{enumerate}
\end{definition}

In \Cref{fig:crossing_links} we are given links $\ell_a=\{a_1,a_6\}$ and $\ell_b=\{b_1,b_4\}$.
Their projection vertex sets are $\mathcal{V}(\ell_a)=(a_1,a_2,\ldots,a_6)$ and
$\mathcal{V}(\ell_b)=(b_1,b_2,\ldots,b_4)$ respectively. Note here that even though some path from
$a_1$ to $a_6$ might pass through the articulation vertex $b_2$, \textit{every} path between the
vertices would not, and hence $b_2$ is not in the projection vertex set of $\ell_a$. Further we can
say that $\ell_a$ and $\ell_b$ cross since their projections cross within the cycle
$\{b_2,a_4,b_3,a_3\}$. For another example of crossing links see \Cref{fig:cac_cross}.
The notion of crossing links enables us to determine whether adding a set of links \(A \subseteq L\)
to the cactus \(G\) makes a given pair of vertices \(u\) and \(v\) 3-edge-connected. This is
formalized in the following proposition.

\begin{figure}
    \centering
    \includegraphics[scale=1.2,page=5]{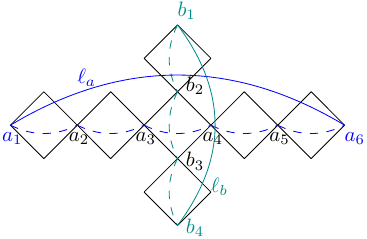}
    \caption{Links $\ell_a$ and $\ell_b$ cross, since their projections $(a_3,a_4)$ and $(b_2,b_3)$ cross within a cycle.}
    \label{fig:crossing_links}
\end{figure}

\begin{restatable}{proposition}{CanConCross}\label{prop:CacConCross}
  If we add a set of links \(A \subseteq L\) to a cactus \(G\), then a pair of distinct vertices
  \((u, v)\) are \(3\)-edge-connected in \(G \cup A\) if and only if there exists a sequence of
  links in \(A\) such that  
    (i) $u$ lies in the projection set of the first link,  
    (ii) $v$ lies in the projection set of the last link, and  
    (iii) every two consecutive links in the sequence cross.  
We refer to such a sequence as a \emph{sequence of crossing links from $u$ to $v$}.
\end{restatable}

We prove \Cref{prop:CacConCross} in \Cref{app:connectedCrossing}.

%%%%%%%%%%%%End of Section%%%%%%%%%

\section{An \texorpdfstring{$O(\log n)$}{O(log n)}-Competitive Algorithm for Online CAP} \label{sec:onlineCAP}

In this section, we obtain an $O(\log n)$-competitive algorithm for CAP, thereby proving
\Cref{thm:online-cap}. From \Cref{thm:cap-from-cacap}, we know that online CAP reduces to online
CacAP. Thus, it suffices to prove the following theorem.

\begin{theorem}\label{thm:logncacap}
    There exists an $O(\log n)$-competitive online algorithm for the online Cactus
    Augmentation Problem (CacAP).
\end{theorem}

To prove~\Cref{thm:logncacap}, we first describe \textit{structured instances} which satisfy some
properties that can be assumed without loss of generality. In order to define structured instances, we first need to define a decomposition tree.

\subsection*{Decomposition Tree}
To each cactus, we associate a rooted tree, called its \textit{decomposition tree}.

To construct a decomposition tree $T$ of a cactus $G$,
first choose an arbitrary vertex $r\in V(G)$ which will become the root of the tree $T$. The vertex set of the tree we construct will be $V(T) = \mathcal{A}\cup\{r\}$, where $\mathcal{A}\subseteq V$ is the set of all
articulation vertices in the cactus $G$. We will construct $T$ incrementally. Initially, $V(T)=\emptyset.$
\begin{itemize}
  \item Insert $r$ into $V(T)$ and mark it visited.
  \item While there exists a cycle $c$ in $G$ that has only one visited vertex and the root of the
      cycle $c$ has not been identified:
  \begin{itemize}
      \item Denote the visited vertex to be the root of the cycle, $r_c$.
      \item Add all articulation vertices in $c$ other than $r_c$ to $V(T)$ and mark them as visited. Furthermore, connect these articulation vertices in the cycle $c$ to $r_c$ by creating new edges
          in $E(T)$.
  \end{itemize}
\end{itemize}

It is straightforward to verify that the above procedure visits all articulation vertices and, for
each cycle $c$, identifies a corresponding root $r_c$.

In \Cref{fig:cac_input} we can see a given cactus, and some
feasible set of links in dashed blue lines. In \Cref{fig:cac_tree} the decomposition tree
is constructed, and depicted in red, with the root being the vertex labeled $r$.

Given a tree $T$ and a link $\ell\in {V(T) \choose 2}$, we say that $\ell$ \textit{covers} an edge of the tree $T$, if adding the link $\ell$ to $T$ ensures that there are
two edge-disjoint paths between the end points of the edge. Also, $\ell$ is said to be an \textit{uplink} if one end point of $
\ell$ is an ancestor of
the other on the tree $T$.

We are now ready to define a Structured instance for online CacAP.

\begin{definition}[Structured instance]
    An instance of online CacAP is said to be a \textit{structured instance} if it satisfies the
    following properties. Below, the decomposition tree refers to the one associated with the cactus provided in the instance.   
    \begin{enumerate}
        \item Every online request lies within a cycle (i.e., each request specifies two vertices of a cycle).
        \item Every link in the instance either is between two vertices of the decomposition tree
            or lies within a cycle.
        \item All links between vertices of the decomposition tree are uplinks.
        \item Every edge of the decomposition tree is covered by a link in the input.
    \end{enumerate}
\end{definition}

Now, the proof of \cref{thm:logncacap} immediately follows from the following two lemmas, which we prove in the following subsections.

\begin{lemma}\label{lem:structuring}
If there exists an $O(\log n)$-competitive online algorithm for structured instances of CacAP, then there exists an $O(\log n)$-competitive algorithm for CacAP.
\end{lemma}

\begin{lemma}\label{lem:solvingStructured}
There exists an $O(\log n)$-competitive online algorithm for structured instances of CacAP.
\end{lemma}

\subsection{Structuring the Input Instance}

We first prove \cref{lem:structuring}.

\begin{proof}

We claim that the four properties of a structured instance can be assumed with only a constant factor loss in the cost of
the optimal solution. We sketch a brief proof for each of the four properties.
    
\begin{enumerate}
    \item{Every online request lies within a cycle:} When a request $(u,v)$ arrives, $3$-edge-connecting $(u,v)$ requires $3$-edge-connecting
        every pair of vertices in its projection vertex set $\mathcal{V}((u,v)) = (u, c_1, \dots,
        c_k, v)$ by \Cref{prop:transitivity}.  Equivalently, the adversary may instead present the
        sequence of requests $(u, c_1), (c_1, c_2), \dots, (c_k, v)$, each confined to a single
        cycle, potentially making it harder for our algorithm to maintain a good competitive ratio.
        Formally, whenever a request $(u,v)$ arrives, our algorithm processes it by sequentially
        handling the requests $(u, c_1), (c_1, c_2), \dots, (c_k, v)$.  Notice, however, that the
        offline optimum remains unchanged in both cases.

    \item{Every link in the instance either is between two vertices of the decomposition tree
            or lies within a cycle:} Consider any link $\ell = (u,v)$ on
        the rooted cactus.  If $u$ and $v$ lie in the same cycle of the input cactus, then $\ell$ trivially
        satisfies the property.
        Otherwise, let the projection vertex set of $\ell$ be  
        \[ \mathcal{V}(\ell) = (u, c_1, \dots, c_k, v), \]  where $k\geq 1$.
        We know that any pair of consecutive vertices in this sequence lie on the same cycle and the $c_i$s are articulation vertices. So we
        replace the link $\ell$ with at most three links as follows: the head $H(\ell)=\{u,c_1\}$, the middle $M(\ell) = \{c_1,c_k\}$, and the tail $T(\ell)=\{c_k, v\}$. If $c_1=c_k$, $M(\ell)$ does not exist. Clearly, $H(\ell)$, $M(\ell)$, and $T(\ell)$ satisfy the property of being either within a cycle, or between two vertices of the decomposition tree (as all articulation vertices are on the decomposition tree).

        Since $H(\ell)$ crosses $M(\ell)$, and $M(\ell)$ crosses $T(\ell)$, by
        \Cref{prop:CacConCross} they serve the same purpose as $\ell$ in 3-edge-connecting requests.
        Thus, we can safely do all such replacements, incurring at most a factor-3 loss in the size of the optimal solution.

    \item{All links between vertices of the decomposition tree are uplinks:} Consider any link $\ell=(u,v)$ between vertices of the decomposition tree. If it is
        not an uplink, let $w$ be the lowest common ancestor of $u$ and $v$ on the decomposition
        tree.
        Replace $\ell$ with two links: $(u,w)$ and $(w,v)$.

        Again, it is easy to see by
        \Cref{prop:CacConCross}, that $\{u,w\}$ and $\{w,v\}$ serve the same purpose as $\ell$ in 3-edge-connecting requests.
        Furthermore, by repeatedly doing this replacement, we ensure that every link between vertices of the decomposition tree is
        an uplink, and for this step 
        we incur at most a factor-2 loss in the size of the optimal solution.
        For an example of splitting a link into uplinks, see \Cref{fig:cac_uplink}.

\item{Every edge of the decomposition tree is covered by a link in the input:} If there exists an edge $e = (u,v)$ on the decomposition tree $T$ that is not
        covered by any link, we break the cactus into two smaller cactus graphs by separating out an
        appropriate sub-cactus and then solve these instances independently.  

        More specifically, suppose the endpoints of $e$, $u$ and $v$, lie on the same cycle $C$ and
        $u$ is the parent of $v$ in the tree $T$.  In this case, we split the cactus
        at $v$ so that the subtree of $T$ rooted at $v$ becomes the tree corresponding to the
        separated sub-cactus. Note that $u$ continues to be an articulation vertex, and hence will
        still be a vertex in the \textit{parent} tree. Meanwhile $v$ need not be an articulation
        vertex anymore.  We can repeatedly apply this splitting procedure until the assumption is
        satisfied for every edge of $T$.

        Since each request is confined to a single cycle, no request has endpoints in two different
        obtained cacti. Furthermore, there cannot be any link with endpoints in the two obtained cactus graphs; since any such link is now necessarily between two articulation vertices (from (2) above), and is an uplink (from (3) above), it definitely covers the tree edge $\{u,v\}$, a contradiction. 
        Now, each of the resulting cactus graphs contains fewer vertices, applying the
        $O(\log n)$-competitive algorithm in parallel to both the instances yields an overall competitiveness at most the
        maximum of the two, which remains $O(\log n)$.  We can now repeatedly split the obtained cactus graphs further if any of their decomposition trees contain an uncovered edge. 
        \qedhere
\end{enumerate}
\end{proof}

\begin{figure}
    \centering
    \begin{subfigure}{0.48\textwidth}
        \centering
        \includegraphics[scale=0.7,page=1]{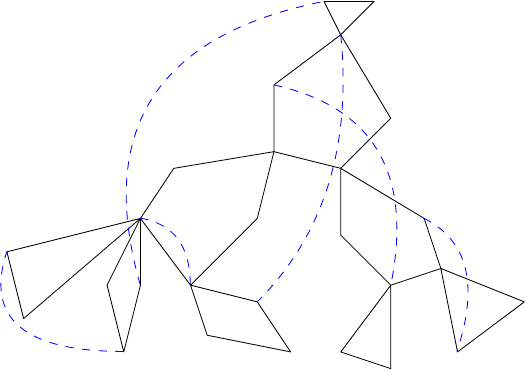}
        \caption{Cactus and links forming a CacAP instance.}
        \label{fig:cac_input}
    \end{subfigure}
    \begin{subfigure}{0.48\textwidth}
        \centering
        \includegraphics[scale=0.7,page=2]{fig_cactus.pdf}
        \caption{Tree obtained from articulation vertices.}
        \label{fig:cac_tree}
    \end{subfigure}
    \begin{subfigure}{0.48\textwidth}
        \centering
        \includegraphics[scale=0.7,page=3]{fig_cactus.pdf}
        \caption{Crossing Links.}
        \label{fig:cac_cross}
    \end{subfigure}
    \begin{subfigure}{0.48\textwidth}
        \centering
        \includegraphics[scale=0.7,page=4]{fig_cactus.pdf}
        \caption{Splitting a link into uplinks.}
        \label{fig:cac_uplink}
    \end{subfigure}
    \caption{Example instance of Cactus Augmentation Problem.}
    \label{fig:cactus}
\end{figure}

\subsection{Solving Structured Instances} All that is left to prove is~\cref{lem:solvingStructured}, i.e., obtaining an $O(\log n)$-competitive algorithm for structured instances.
We
make use of the following two theorems to give an algorithm that can solve such
instances.
\taponline*

\begin{restatable}{theorem}{unwtcycap}
 \label{thm:unwtcycap}
There exists a deterministic $O(\log n)$-competitive algorithm for
online {\sc Cycle Augmentation Problem} (CycAP).
\end{restatable}

\Cref{thm:online-tap} follows from the work of Naor et al.~\cite{NaorUW22}. We include an alternate
proof in~\Cref{sec:onlineTap}. \Cref{thm:unwtcycap} follows from a reduction to the online Steiner
forest problem, and we give the proof in \Cref{sec:logsqnWCAP}. 
 
We will design an algorithm $\ALG$ for structured instances of
online CacAP and prove that it is $O(\log n)$-competitive. Let $\ALG_{T}$ be a $(c_t \log
n)$-competitive algorithm for online TAP and $\ALG_{R}$  be a $(c_r  \log n)$-competitive algorithm
for online CycAP promised by \Cref{thm:online-tap,thm:unwtcycap}, respectively, where $c_t$ and
$c_r$ are absolute positive constants. These algorithms will act as subroutines of our main
algorithm $\ALG$.

Let the input to our algorithm $\ALG$ be a cactus $G=(V,E)$ and the set of links $\mathcal{L}$.
Given $G$ and $\mathcal{L}$, we create an online CycAP instance for each cycle in $G$ and a single
online TAP instance. Each cycle instance is handled by an independent copy of $\ALG_R$, while the
TAP instance is handled by $\ALG_T$. Upon the arrival of a request, we carefully generate
corresponding subrequests for the relevant subroutines.  The idea is that $\ALG_R$ handles local
optimizations within individual cycles, while $\ALG_T$ manages global, long-range optimizations
across the cactus structure.

\paragraph{The TAP instance.}  
The TAP instance is defined on the decomposition tree of the cactus. And the set of links
$\mathcal{M}$ for this instance, consists of all
the links which have both their end points on vertices of the decomposition tree.

\paragraph{The CycAP instances.}  
Each cycle in $G$ is treated as a separate CycAP instance. For the instance corresponding to a cycle
$C$, we define its link set as  
\[ \mathcal{L}(C) = \bigcup_{\ell \in \mathcal{L}} \proj_C(\ell). \]  
Notice that $\mathcal{L}(C)$ consists of the projections of all links in $\mathcal{L}$ onto the
cycle $C$.

\paragraph{Algorithm.}

With the above setup, we can now describe our algorithm $\ALG$.  
When an uncovered request $(s_i, t_i)$ arrives in some cycle $c$:  

\begin{enumerate}
    \item TAP phase:  For each vertex $v \in \{s_i, t_i\}$ that is on the decomposition
        tree, forward the request $(r_c, v)$ to the online TAP algorithm $\ALG_T$ if it is not
    already covered. Add the link
        selected by $\ALG_T$ to our solution.
    \item CycAP phase:  If $(s_i, t_i)$ remains uncovered, forward it to the online CycAP
       algorithm $\ALG_R^c$ corresponding to cycle $c$.  For every projected link $\proj_c(\ell)$
       selected by $\ALG_R^c$, include the corresponding original link $\ell \in \mathcal{L}$ in our
       solution.  
\end{enumerate}
Notice that $\ALG$ always succeeds in serving any request whenever a feasible solution exists.  
This holds because every request lies entirely within a single cycle, and for each cycle instance,  
the online algorithm has access to the projections of all links in $\mathcal{L}$ onto that cycle.

\subsection{Analysis via a Charging Argument}

We now analyze the performance of $\ALG$.

Let $\COPT$ denote the offline optimal solution for the given structured cactus instance.  
Fix a cycle $c$ in the cactus $G$, and let $\COPT_c$ be the set of projections of $\COPT$ onto cycle $c$.  
We partition  
\[
\COPT_c = \COPT_c^{W} \sqcup \COPT_c^{M},
\]  
where $\COPT_c^{W}\subseteq \mathcal{L}$ consist of all links in $\mathcal L$ that lie within the cycle $c$, and $\COPT_c^{M}$ consist of projections of the other links onto the cycle $c$. Intuitively, links in $\COPT_c^M$ consist of projections that are in the `middle' of other links. Notice that one endpoint of each link in $\COPT_c^M$ is necessarily $r_c$, the root of the cycle $c$.

Let $\OPT_c$ denote the offline optimal solution for all the requests that our algorithm $\ALG$ forwarded to the cycle-specific algorithm $\ALG_R^c$.
For simplicity of exposition, we will refer to all elements of $\OPT_c$ and $\COPT_c$ as links rather than projections of links. 

Let $\TOPT$ denote the optimal solution for all the requests that our algorithm $\ALG$ forwarded to the tree algorithm $\ALG_T$.  

We will prove the following two bounds on the sizes of $\COPT^c$ and $\TOPT$.

\begin{lemma}\label{lem:cyclecharge}
    $17 \cdot |\COPT_c^{W}| \ge |\OPT_c|$.
\end{lemma}

\begin{lemma}\label{lem:treecharge}
    $2 \cdot |\COPT| \ge |\TOPT|$.
\end{lemma}

Before proving the above lemmas, we first show how they suffice to complete the proof of \Cref{lem:solvingStructured}
\begin{proof}[Proof of \Cref{lem:solvingStructured}]
    Since each link in $\COPT$ which is within a cycle is part of only one cycle, we have 
    \[|\COPT| \geq  \sum_c|\COPT_c^{W}| .\] 

    Also the cost of our algorithm is at most
    \begin{align*}
        |\ALG| &\le \left(\sum_c c_r\log n_c\cdot|\OPT_c|\right) + c_t\log n\cdot|\TOPT| \\
               &\le (c_r+c_t)\log n\left( \sum_c|\OPT_c| + |\TOPT|\right) \\
               &\le (c_r+c_t)\log n\left( 17\sum_c |\COPT_c^{W}| + 2\cdot|\COPT| \right)
                \tag{From \Cref{lem:cyclecharge} and \ref{lem:treecharge}} \\
               &\le 19(c_r+c_t)\log n\cdot|\COPT| =O(\log n)\cdot|\COPT|.
    \end{align*}
    Hence we see that the algorithm returns a solution that is at most $O(\log n)$ times the optimal
    offline solution, giving us the required competitive ratio.
\end{proof}

\subsubsection{Bound for a Cycle Instance}
We now prove~\cref{lem:cyclecharge}, i.e.,
    $17 \cdot |\COPT_c^{W}| \ge |\OPT_c|$.

\begin{proof}
Note that $\COPT_c$ is a feasible solution for the cycle instance $c$, but $\COPT_c^{W}$ alone need
not be feasible.  Therefore, consider the set $\SOL = \COPT_c^{W} \sqcup \COPT_c^{M'}$, where  
$\COPT_c^{M'} \subseteq \COPT_c^{M}$ is a minimal subset such that $\SOL$ covers all requests assigned to the cycle $c$.  
Thus, we have
\[
    |\COPT_c^{M'}| + |\COPT_c^{W}| \ge |\OPT_c| .
\]

We will show via a charging argument that $|\COPT_c^{M'}| \le 16\,|\COPT_c^{W}|$, which from above will immediately give us the desired bound: 
\[ 17\,|\COPT_c^{W}| \;\ge\; |\OPT_c|.\]

Now to show $|\COPT_c^{M'}| \le 16\,|\COPT_c^{W}|$, construct a bipartite graph whose left nodes correspond to all links in $\COPT_c^{M'}$, and
whose right nodes correspond to all requests $(s_i, t_i)$ sent to cycle $c$.  To define edges in
this bipartite graph, fix for each request $(s_i, t_i)$ a shortest $s_i$–$t_i$ path in $c$, meaning
a sequence of crossing links from $\SOL$ where the first link is incident to $s_i$, the last link is
incident to $t_i$ (such a
sequence of links exists by \Cref{prop:CacConCross}).

Observe that each link in $\COPT_c^{M'}$ belongs to the middle portion of some original link and is
therefore incident on the root $r_c$ of cycle $c$ (since structured instances have only uplinks between two vertices that are not within a cycle).
Consequently, any shortest $s_i$–$t_i$ path can include at most two links from $\COPT_c^{M'}$, as
otherwise it would revisit $r_c$ multiple times, contradicting our assumption that the path is shortest.  

In the bipartite graph, we add an edge between a link $e_i \in \COPT_c^{M'}$ (left node) and a
request $(s_j, t_j)$ (right node) if $e_i$ lies on the fixed shortest $s_j$–$t_j$ path.  By
construction, the degree of any right node (request) is at most two, while every left node has
degree at least one (by the minimality of $\COPT_c^{M'}$). It is easy to see that such a bipartite graph with these degree bounds has a matching $D$ such that

\[
    2|D| \ge |\COPT_c^{M'}| .
\]

Finally, note that a request $(s,t)$ between two articulation vertices is never assigned to a cycle:
such requests are instead handled in the TAP phase, which processes $(s,r_c)$ and $(t,r_c)$.  Since
the TAP instance is assumed feasible, it $3$-edge-connects $(s,r_c)$ and $(r_c,t)$, and hence also
$(s,t)$.  Therefore, without loss of generality, we may assume that each $s_i$ is a non-articulation
vertex.

Now, consider the set of terminals of the requests matched in $D$ above (i.e., the $s_i$’s and
$t_i$’s) and form a graph where each request is represented as an edge between its two terminals.
This graph cannot contain a cycle, since any edge that would close a cycle corresponds to a request
that is already satisfied and therefore would never have been sent to the CycAP instance.  Let there
be $a$ terminals labeled $s_i$ (non-articulation vertices) and $b$ terminals labeled $t_i$. Since this graph is a forest with
$|D|$ edges, we have $a + b \ge |D| + 1\geq |D|$.

    \begin{description}
        \item[Case 1: $a\ge|D|/2$.] Consider a link $e_i$ and a request $(s_i, t_i)$ that are matched in $D$.
             Since $s_i$ is a non-articulation vertex, the first link on the
            $s_i$--$t_i$ path must lie in $\COPT_c^{W}$, because both endpoints of any link in
            $\COPT_c^{M'}$ are articulation vertices.  Thus, there are at least $|D|/4$ distinct links in $\COPT_C^{W}$, which act as the first links in the $s_i$--$t_i$ paths above. Thus, $|D|\leq 4|\COPT_c^W|$.

        \item[Case 2: $b \geq |D| /2$, and more than $|D|/4$ of the $t_i$'s are
            non-articulation vertices.] Consider a link $e_i$ and a request
            $(s_i, t_i)$ that are matched in $D$ and $t_i$ is a non-articulation vertex.  Since $t_i$ is not an articulation vertex, the last link on the $s_i$--$t_i$ path must belong to $\COPT_c^{W}$. Thus, there are at least $|D|/8$ distinct links in $\COPT_C^{W}$, which act as the last links in the $s_i$--$t_i$ paths above. Thus, $|D|\leq 8|\COPT_c^W|$.

        \item[Case 3: $b \geq |D| /2$, and more than $|D|/4$ of the $t_i$'s are articulation
            vertices.]
            Consider all requests $(s_i,t_i)$ mapped in $D$, where $t_i$ is an articulation vertex.
            In this situation we show that all $s_i$ in these requests must be distinct. To see why, suppose
            $s_i = s_j$ for some $j$ arriving after $i$.  When the request $(s_j, t_j)$ arrives, the
            algorithm would already have served the requests  $(r_c, t_i)$, $(s_i, t_i)$, and $(r_c,
            t_j)$.  Hence, $(s_j, t_j)$ would already be connected at the time of arrival and would
            therefore never be sent to the cycle instance. Thus, there are at least $|D|/4$ distinct $s_i$, and similar to case 1 above, there are at least $|D|/8$ distinct links in $\COPT^W_c$, giving us $|D|\leq 8|\COPT^W_C|$.

        We have shown from the above charging argument that $|D| \le 8\,|\COPT_c^{W}|$, which implies  
$|\COPT_c^{M'}| \le 16\,|\COPT_c^{W}|$, as desired.  \qedhere
    \end{description}
\end{proof}

\subsubsection{Bound for the Tree Instance} Finally, we prove~\cref{lem:treecharge}. i.e.,
    $2 \cdot |\COPT| \ge |\TOPT|$.
\begin{proof}

We will construct a feasible solution $\TOPT'$ for the TAP instance—of size at most $2|\COPT|$. We
partition $\COPT = \COPT^1 \cup \COPT^2$, where $\COPT^1$ consists of all links in $\COPT$ each of
which covers at least one tree request forwarded to $\ALG_T$. Add all requests in $\COPT^1$ to
$\TOPT'$.

Now, consider the set of all tree requests, of cardinality $q$, generated for a cycle $c$: $\{\{s_1,
r_c\},\dots, \{s_q, r_c\}\}$ that are not already covered by $\COPT^1$. A request $\{s_i,r_c\}$ was
generated due to some \textit{parent} request $\{s_i,t_i\}$ for the cactus instance. Thus, $\COPT$
must have a sequence of crossing links covering $\{s_i,t_i\}$, where $s_i$ is in the projection
vertex set of the first link of the sequence. If the first link in the sequence is $\{s_i,r_c\}$ or
if it is not within the cycle $c$, then it must be an uplink covering $\{s_i,r_c\}$, a contradiction
to our assumption that $\{s_i, r_c\}$ is not already covered by $\COPT^1$. Thus, the first link in
the sequence is $\{s_i, w_i\}\in \COPT^2$, which is within the cycle $c$ and $w_i\neq r_c$. Thus,
there are at least $q/2$ links in $\COPT^2$ within the cycle $c$ that serve as the first link in the
sequence of crossing links to cover the $q$ requests $\{s_i,t_i\}$ corresponding to the $q$ tree
requests $\{s_i,r_c\}$. In $\TOPT'$ add $q$ links each covering one $\{s_i,r_c\}$  tree request
(which exist as each edge of the decomposition tree in a structured instance is covered by some
link). Do this for each cycle $c$. Thus, clearly, $|\TOPT'| \leq |\COPT^1| + 2|\COPT^2| \leq
2|\COPT|$. Now since $\TOPT'$ is a feasible solution for the tree instance and $\TOPT$ is an optimal
solution, $|\TOPT|\leq|\TOPT'|$, and we are done.
\end{proof}

%%%%%%%%%%%%End of Section%%%%%%%%%

\section{An \texorpdfstring{$\Omega(\log n)$}{Ω(log n)} Lower Bound for TAP under Random Order Arrivals}
\label{sec:ROlowerbound}
In this section, we prove \Cref{thm:random-order-lower}. More specifically, we construct an instance of TAP in which the requests arrive in random order, and show that every algorithm for TAP has competitive ratio $\Omega(\log n)$.

Consider a complete binary tree $B$ with $n$ vertices. We now subdivide the edges of the tree in the
following way to obtain a tree $T$. Edges at level $1$ in $B$ are replaced by paths of length $n/2$
in $T$, edges at level $2$ in $B$ are replaced by paths of length $n/4$ in $T$, and so on. In
general, edges at level $i$ in $B$ are replaced by paths of length $n/2^i$ in $T$. The branching
vertices in $T$ are in one to one correspondence with the vertices of $B$. An edge in $T$ is said to
be in level $i$ if it appears between two consecutive branching vertices whose corresponding levels
in $B$ are $i-1$ and $i$. The number of vertices in $T$ are $O(n\log n)$ and the number of leaves in
$T$ and $B$ are the same.

The tree augmentation instance will consist of the tree $T$ and for each leaf $\ell$, $(r,\ell)$
will be a link, where $r$ is the root of $T$. The online requests consist of all edges on a root to
leaf path, where the leaf is chosen uniformly at random from the leaf. Equivalently, the requests
consists of all edges on a uniform random walk starting at the root and going down the tree.

The requests are revealed to an online algorithm in a uniformly random order.  Without loss of
generality, an online algorithm when presented an uncovered request will pick precisely one link
covering that request. 

We call a request $e$ at level $i$ a \emph{fresh} request, if it is requested before any other
requests in levels $\{i,i+1,\ldots,\log n\}$. Note that the algorithm never picks a link on a
request which is not a fresh request (uncovered implies fresh). Furthermore, on a fresh request the
algorithm picks an edge with probability at least $\frac 1 2$ (the algorithm has picked precisely
one link that covers all the previous requests).

Let $S$ denote the number of fresh requests. Let $X_i$ be an indicator random variable which is $1$
iff on the $i^\text{th}$ fresh request the algorithm picks an edge. Thus the cost of the algorithm
is precisely $X_1 + \cdots + X_S$.

Now the probability that a request from level $i$ is fresh is at least $\frac 1 2$, as the total
number of requests in levels $\{i+1, \cdots, \log n\}$ is one less than the number of requests in
level $i$. Thus, $\mathbb E[S] \geq \frac {\log n} 2$.

Now, we can compute the expected cost:
 \begin{align*}  \mathbb E[X_1 + \cdots + X_S]= \sum_{i=1}^{\log n}\Pr[S=s]\mathbb E[X_1 + \cdots + X_s| S = s]
    \geq \sum_{i=1}^{\log n}\Pr[S=s] \frac s 2
    = \frac 1 2 \mathbb E[S]
    \geq \frac 1 4 \log n.
\end{align*}

%%%%%%%%%%%%End of Section%%%%%%%%%

\section{%
\texorpdfstring{$O(\log^2 n)$}{O(log^2 n)}- and
\texorpdfstring{$O(\log n)$}{O(log n)}-Competitive Algorithms
for Online WCAP and Online CycAP} \label{sec:logsqnWCAP}

In this section, we present an $O(\log^2 n)$-competitive algorithm for WCAP, proving~\Cref{thm:weighted-cap}. By~\Cref{thm:cap-from-cacap}, it suffices to design an $O(\log^2 n)$-competitive algorithm for the weighted cactus augmentation problem. We establish the following result.

\begin{restatable}{theorem}{logsqcacap}\label{thm:logsqcacap}
There exists an $O(\log^2 n)$-competitive online algorithm for the online
{\sc Weighted Cactus Augmentation Problem} (WCacAP).
\end{restatable}

    We then improve this competitive ratio for the special case in which the underlying cactus is a cycle and the links are unweighted, proving~\Cref{thm:unwtcycap}.

\unwtcycap*

Our proof of \Cref{thm:logsqcacap} proceeds via a reduction to the Online Node-Weighted Steiner Forest (NWSF) problem and the application of a known online NWSF algorithm. For the special case of unweighted cycle augmentation, we show that the resulting NWSF instance admits an approximation via a reduction to the Online Edge-Weighted Steiner Forest (EWSF) problem, for which a better competitive algorithm is known, thereby yielding \Cref{thm:unwtcycap}. We first recall the NWSF and EWSF problems.

\begin{definition}[Node-Weighted Steiner Forest]
Given an undirected graph $G = (V,E)$, a cost function $c : V \to \mathbb{R}_{\ge 0}$ on the vertices, and a set of terminal pairs $\mathcal{T} \subseteq V \times V$, find a minimum-cost subset of vertices $S \subseteq V$ such that for each $(s,t) \in \mathcal{T}$, the induced subgraph $G[S]$ contains a path between $s$ and $t$.
\end{definition}

\begin{definition}[Edge-Weighted Steiner Forest]
Given an undirected graph $G = (V,E)$, a cost function $c : E \to \mathbb{R}_{\ge 0}$ on the edges, and a set of terminal pairs $\mathcal{T} \subseteq V \times V$, find a minimum-cost subset of edges $F \subseteq E$ such that for each $(s,t) \in \mathcal{T}$, the subgraph $(V,F)$ contains a path between $s$ and $t$.
\end{definition}

In the online versions of the above problems, the set of terminal pairs is not given upfront; instead, request pairs arrive one by one in an online fashion. Upon the arrival of a request pair, an online algorithm must irrevocably select vertices or edges so as to connect the requested pair using the vertices or edges chosen so far.

\subsection{Reduction to Online NWSF}

Our reduction from Online WCacAP to Online NWSF is inspired by earlier reductions of the offline
Connectivity Augmentation Problem to the Steiner Tree problem
\cite{BasavarajuFGMRS14,BGJ20}.

Given an instance of WCacAP with cactus $G$ and a set of links $L$, we construct an instance of NWSF,
denoted by $H$, where the vertex set is $V(H) = V_H \sqcup T_H$, as follows.
\begin{itemize}
  \item \textbf{Steiner nodes.}
  For every link $\ell = \{u,v\} \in L$ with cost $w_\ell$, we create a Steiner node $v_\ell \in V_H$
  with cost $c(v_\ell) = w_\ell$.

  \item \textbf{Terminal nodes.}
  For every vertex $v \in V(G)$, we create a terminal node $v \in T_H$ of zero cost.

  \item \textbf{Incidence edges.}
  For every link $\ell \in L$, and for every vertex $u$ in the projection vertex set of $\ell$, we add
  the edge $\{v_\ell, u\}$ to $E(H)$.

  \item \textbf{Crossing edges.}
  For every pair of crossing links $\ell_1$ and $\ell_2$, we add the edge
  $\{v_{\ell_1}, v_{\ell_2}\}$ to $E(H)$.
\end{itemize}

When a WCacAP request $(s,t)$ arrives, it is translated into an NWSF request to connect the
corresponding terminal nodes $s$ and $t$ in the graph $H$. We handle these requests using an online
algorithm for NWSF. For each Steiner node $v_\ell$ selected by the online NWSF algorithm, we purchase
the corresponding link $\ell$ in the WCacAP instance. An example of this reduction is illustrated in
\Cref{fig:NWSFreduction}.

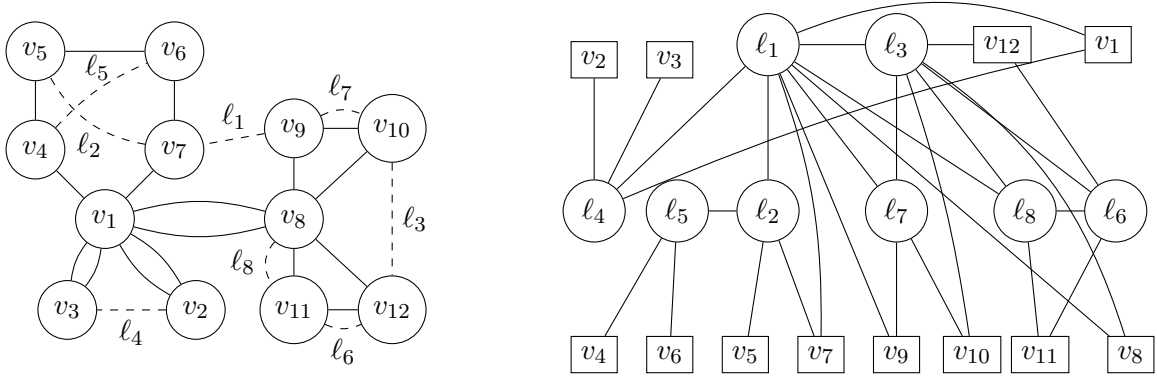
\begin{figure}
\hspace{15pt}
\begin{tikzpicture}
    % Define the style for 'state' nodes (for v_i nodes)
    \tikzstyle{state}=[fill=white, draw, minimum size=10pt, circle]
    
    % Define the nodes
    \node[state] (1) {$v_1$};
    \node[state, below right of=1, node distance=1.7cm] (2) {$v_2$};
    \node[state, left of=2, node distance=1.7cm] (3) {$v_3$};
    \node[state, right of=1, node distance=2.5cm] (8) {$v_8$};
    \node[state, above of=8, node distance=1.2cm] (9) {$v_9$};
    \node[state, right of=9, node distance=1.3cm] (10) {$v_{10}$};
    \node[state, below of=8, node distance=1.2cm] (11) {$v_{11}$};
    \node[state, right of=11, node distance=1.3cm] (12) {$v_{12}$};
    \node[state, above right of=1, node distance=1.3cm] (7) {$v_7$};
    \node[state, above left of=1, node distance=1.3cm] (4) {$v_4$};
    \node[state, above of=7, node distance=1.3cm] (6) {$v_6$};
    \node[state, above of=4, node distance=1.3cm] (5) {$v_5$};

    % Define the paths for edges
    \path (1) edge[bend right=15] (8);
    \path (1) edge[bend left=15] (8);
    \path (1) edge[bend left=15] (2);
    \path (1) edge[bend right=15] (2);
    \path (1) edge[bend left=15] (3);
    \path (1) edge[bend right=15] (3);
    \path (8) edge (9);
    \path (8) edge (10);
    \path (10) edge (9);
    \path (8) edge (11);
    \path (8) edge (12);
    \path (12) edge (11);
    \path (1) edge (4);
    \path (4) edge (5);
    \path (5) edge (6);
    \path (6) edge (7);
    \path (7) edge (1);

    % Define dashed edges
    \path [dashed] (9) edge [bend right=0] node[above] {$\ell_1$}(7);
    \path [dashed] (7) edge[bend left=25] node[below] {$\ell_2$} (5);
    \path [dashed] (4) edge[bend left=12] node[above] {$\ell_5$} (6);
    \path[dashed] (2) edge node[below] {$\ell_4$}(3);
    \path[dashed] (10) edge node[right]{$\ell_3$}(12);
    \path[dashed] (11) edge[bend right=25] node[below] {$\ell_6$} (12);
    \path[dashed] (9) edge[bend left=25] node[above] {$\ell_7$} (10);
    \path[dashed] (8) edge[bend right=40] node[left] {$\ell_8$} (11);
\end{tikzpicture}
\hfill
\begin{tikzpicture}
    % Define the style for 'state' nodes (for \ell_i nodes)
    \tikzstyle{state}=[fill=white, draw, minimum size=5pt, circle]
    
    % Define the \ell_i nodes in circles
    \node[state] (l1) {$\ell_1$};
    \node[state, right of=l1, node distance=1.7cm] (l3) {$\ell_3$};
    \node[state, below of=l1, node distance=2.2cm] (l2) {$\ell_2$};
    \node[state, left of=l2, node distance=1.2cm] (l5) {$\ell_5$};
    \node[state, left of=l5, node distance=1.1cm] (l4) {$\ell_4$};
    \node[state, right of=l2, node distance=1.7cm] (l7) {$\ell_7$};
    \node[state, right of=l7, node distance=1.7cm] (l8) {$\ell_8$};
    \node[state, right of=l8, node distance=1.2cm] (l6) {$\ell_6$};

    % Define the rectangle nodes (for v nodes)
    \tikzstyle{vstate}=[fill=white, draw, minimum size=8pt, rectangle]
    \node[vstate, above of=l4, node distance=2cm] (2) {$v_2$};
    \node[vstate, right of=2, node distance=1cm] (3) {$v_3$};
    \node[vstate, below of=l4, node distance=1.9cm] (4) {$v_4$};
    \node[vstate, right of=4, node distance=1cm] (6) {$v_6$};
    \node[vstate, right of=6, node distance=1cm] (5) {$v_5$};
    \node[vstate, right of=5, node distance=1cm] (7) {$v_7$};
    \node[vstate, right of=7, node distance=1cm] (9) {$v_9$};
    \node[vstate, right of=9, node distance=1cm] (10) {$v_{10}$};
    \node[vstate, right of=10, node distance=0.9cm] (11) {$v_{11}$};
    \node[vstate, right of=l3, node distance=1.4cm] (12) {$v_{12}$};
    \node[vstate, right of=11, node distance=1.2cm] (8) {$v_{8}$};
    \node[vstate, right of=12, node distance=1.4cm] (1) {$v_1$};

    % Define edges
    \path (l1) edge[bend left=10] (7);
    \path (l1) edge (9);
    \path (l3) edge[bend left=5] (10);
    \path (l3) edge (12);
    \path (l2) edge (5);
    \path (l2) edge (7);
    \path (l4) edge (2);
    \path (l4) edge (3);
    \path (l5) edge (4);
    \path (l5) edge (6);
    \path (l7) edge (9);
    \path (l7) edge (10);
    \path (l6) edge (12);
    \path (l6) edge (11);
    \path (l1) edge[bend left=0] (8);
    \path (l1) edge[bend left=22] (1);
    \path (l3) edge[bend left=19] (8);
    \path (l4) edge[bend left=5] (1);

    % Define cross edges
    \path (l4) edge (l1);
    \path (l3) edge (l7);
    \path (l3) edge (l1);
    \path(l3) edge (l6);
    \path (l2) edge (l5);
    \path (l1) edge (l7);
    \path (l8) edge (11);
    \path (l8) edge (l6);
    \path (l8) edge (l3);
    \path (l8) edge (l1);
    \path (l1) edge (l2);
\end{tikzpicture}
\hspace{15pt}
\caption{\cite{BGJ20} {\bf (left)} Instance of CacAP, where dashed edges denote links. {\bf (right)} The Corresponding Steiner forest instance, where square nodes denote terminals.}
\label{fig:NWSFreduction}
\end{figure}

The following lemma shows that, for any sequence of requests, a set of links $ L'\subset L$  is feasible for the WCacAP instance if and only if the corresponding Steiner nodes $S_{L'}=\{v_\ell \mid \ell \in L'\}$ is feasible for the associated NWSF instance.

\begin{lemma}\label{lem:reduction_preservation}
A set of links $L' \subseteq L$ ensures three edge-disjoint paths between a pair $(s,t)$ in
$G \cup L'$ if and only if $s$ and $t$ are connected in the induced subgraph
$H[S_{L'} \cup T_H]$, where $S_{L'} := \{v_\ell \mid \ell \in L'\}$.
\end{lemma}

\begin{proof}
By \Cref{prop:CacConCross}, a set of links $L' \subseteq L$ ensures three edge-disjoint paths between
$s$ and $t$ in $G \cup L'$ if and only if there exists a sequence of links in $L'$ such that
$s$ lies in the projection vertex set of the first link, $t$ lies in the projection vertex set of
the last link, and every two consecutive links in the sequence cross.

By construction of the reduced NWSF instance $H$, each link $\ell \in L$ corresponds to a Steiner
node $v_\ell$, and for every pair of crossing links $\ell_1$ and $\ell_2$, the graph $H$ contains the
edge $\{v_{\ell_1}, v_{\ell_2}\}$. Moreover, each terminal node $u \in T_H$ is adjacent in $H$ to all
Steiner nodes corresponding to links whose projection vertex set contains $u$.

Hence, the existence of such a sequence of crossing links in $L'$ from $s$ to $t$ is equivalent to
the existence of a path between terminals $s$ and $t$ in the induced subgraph
$H[S_{L'} \cup T_H]$, where $S_{L'} := \{v_\ell \mid \ell \in L'\}$. This completes the proof.
\end{proof}

Thus, \Cref{lem:reduction_preservation} immediately implies that the above reduction is both
cost- and competitiveness-preserving: each purchased link of cost $w_\ell$ in WCacAP corresponds
one-to-one to a purchased Steiner node of the same cost in NWSF, and vice versa. Consequently, any
competitive or approximation guarantee achieved for the online NWSF instance transfers unchanged
to the online WCacAP instance. Since an $O(\log^2 n')$-competitive algorithm is known for
NWSF~\cite{Borst0V25}, where $n' = |V(H)| = |V_H| + |T_H| = |V(G)| + |L| \le n + \binom{n}{2}$ and
$n = |V(G)|$, we obtain the following theorem.

\logsqcacap*
In the next subsection, we apply the above reduction to the special case of unweighted cycle augmentation.

\subsection{Approximating CycAP using EWSF}

While the above reduction allows us to reduce WCacAP exactly to NWSF, for the special case of
\emph{unweighted} cycle augmentation (CycAP) we further approximate the resulting NWSF instance by reducing
it to an instance of EWSF.

Consider the NWSF graph $H$ obtained from an unweighted CycAP instance consisting of a cycle $G$ and a set of links $L$. We construct an EWSF
instance on the same graph $H$ by assigning unit weight to every edge in $E(H)$. Note that in this
setting, each Steiner node is adjacent to exactly two terminals, since the projection of any link
consists precisely of its two endpoints.

Let a sequence of requests define a set of $k$ distinct terminals. Let $\OPT_N$ denote an optimal
solution to the NWSF instance, consisting of a set of Steiner nodes, with cost $|\OPT_N|$. Let
$\OPT_E$ denote an optimal solution to the corresponding EWSF instance, consisting of a set of
edges, with cost $|\OPT_E|$.

\begin{claim}\label{clm:opt_relation}
$|\OPT_E| \le |\OPT_N| + k - 1$.
\end{claim}

\begin{proof}
Let the subgraph of $H$ induced by $\OPT_N$ together with the $k$ distinct terminals appearing in the
requests consist of $c$ connected components. Observe that it suffices to consider only these $k$
terminals: if a path connecting a request $(s,t)$ uses an intermediate terminal
$u$ not among the $k$ terminals, then this path must contain a subpath of the form
$v_{\ell_1} - u - v_{\ell_2}$. Since both $\ell_1$ and $\ell_2$ have $u$ in their projection vertex
sets, the two links cross, and by construction of $H$ the edge
$\{v_{\ell_1}, v_{\ell_2}\}$ exists. Hence, a shorter path avoiding $u$ can be used instead.

We now construct a spanning tree within each of the $c$ connected components, thereby obtaining a
forest that connects all $k$ terminals. This forest contains
$|\OPT_N| + k - c \le |\OPT_N| + k - 1$ edges and serves all requests. Hence, it constitutes a
feasible solution to the EWSF instance, implying the claimed bound on $|\OPT_E|$.
\end{proof}

\begin{claim}\label{clm:k_bound}
$k \le 2|\OPT_N|$.
\end{claim}

\begin{proof}
Each Steiner node in $\OPT_N$ is adjacent to exactly two terminals, and every terminal among the $k$ terminals appearing in the requests must be
adjacent to at least one Steiner node in $\OPT_N$ in order to be connected. Therefore, the total number of
distinct terminals satisfies $k \leq 2|\OPT_N|$.
\end{proof}

We now prove the main theorem of this subsection.

\begin{proof}[Proof of \Cref{thm:unwtcycap}]
We apply a known $O(\log n')$-competitive online algorithm for EWSF~\cite{ImaseW91} to the constructed
EWSF instance $H$, where $n' = |V(H)| = |V(G)| + |L| \le n + \binom{n}{2}$. Let the competitive ratio
of this algorithm be $c_e \log n'$ for some absolute constant $c_e$. The algorithm outputs an edge
set $\ALG_E$ with
\[
|\ALG_E| \le c_e \log n' \cdot |\OPT_E|.
\]

Our algorithm for NWSF (and hence for CycAP) selects both endpoints of each edge chosen in $\ALG_E$.
Therefore, the number of Steiner nodes selected, denoted by $\ALG_N$, satisfies
\[
|\ALG_N| \le 2|\ALG_E|.
\]

Using \Cref{clm:opt_relation}, we obtain
\[
\frac{|\ALG_N|}{|\OPT_N|}
\le \frac{2|\ALG_E|}{|\OPT_N|}
\le \frac{2c_e \log n' \cdot |\OPT_E|}{|\OPT_N|}
\le \frac{2c_e \log n' \cdot (|\OPT_N| + k)}{|\OPT_N|}.
\]
Applying \Cref{clm:k_bound}, which gives $k \le 2|\OPT_N|$, we conclude
\[
\frac{|\ALG_N|}{|\OPT_N|}
\le 2c_e \log n' \cdot \frac{3|\OPT_N|}{|\OPT_N|}
= 6c_e \log n'
= O(\log n).
\]

Finally, since the reduction from Online CycAP to Online NWSF is competitiveness-preserving
(\Cref{lem:reduction_preservation}), this yields an $O(\log n)$-competitive online algorithm for the
unweighted Online CycAP. This completes the proof.
\end{proof}

% %%%%%%%%%%%%End of Section%%%%%%%%%

\section{Simple \texorpdfstring{$O(\log n)$}{O(log n)}-Competitive Algorithm for TAP}\label{sec:onlineTap}

In this section, we describe a simple $O(\log n)$-competitive algorithm for TAP, thereby proving Theorem~\ref{thm:online-tap}. The underlying graph is a rooted tree $T$, and without loss of generality, we may assume that all links are uplinks (i.e., one endpoint of a link is an ancestor of the other in $T$), at the cost of increasing the competitive ratio by at most a factor of $2$. Furthermore, without loss of generality, we may assume that the online requests are \emph{edge-requests} (i.e., the two requested vertices are adjacent in $T$). Both of these assumptions are standard and were also used in~\cite{NaorUW22}.

The inspiration for our algorithm comes from a simple $O(1)$-competitive algorithm for online interval set cover (or equivalently, the online path augmentation problem). In the online interval set cover problem, we are given a collection of intervals, and points that need to be covered arrive online.
The algorithm considers a newly arrived uncovered request together with all intervals that cover it. Among these feasible intervals, it selects the interval with the leftmost left endpoint and the interval with the rightmost right endpoint. It can then be shown that, for every interval $I$ in the optimal solution, the algorithm selects at most two intervals that together cover all requests covered by $I$. This proves a competitive ratio of $2$ for the algorithm.

Similarly, for online TAP, we intuitively wish to select a \emph{topmost link} and a \emph{bottommost link} among the links that cover a newly arrived uncovered edge request \(e_i\).
Given a rooted tree, we say that a link \(t_i\) is a topmost link for \(e_i\) if, among all links covering \(e_i\), it covers the maximum number of edges on the path from \(e_i\) to the root. We now describe the \textbf{bottommost link} \(b_i\) for the request \(e_i\).
Starting from \(e_i\), we traverse the tree along links covering \(e_i\) in the direction of increasing distance from the root. Whenever we encounter a branching node at which multiple branches still contain valid links (i.e., links covering \(e_i\)), we continue along the branch whose subtree contains the largest number of vertices. The traversal stops once only a single valid link remains, or when a leaf is reached. We call the resulting link the bottommost link for the request \(e_i\). In both definitions above, ties may be broken arbitrarily.

The definition of the bottommost link that we use is based on the number of vertices contained in a subtree. The high-level intuition is that every incorrect choice made by the algorithm (i.e., selecting a non-optimal link) eliminates a large portion of the tree that an adversary could otherwise exploit to further force the algorithm into making additional incorrect choices. We therefore bound the number of such incorrect choices in order to obtain a good upper bound on the competitive ratio.

\subsection{Algorithm}
As mentioned earlier, the algorithm is very simple. For each uncovered request $e$, we select a topmost link and a bottommost link that cover $e$. A formal description of the algorithm is provided in~\Cref{alg:rotap}.

\begin{algorithm}
    \caption{Online algorithm for unweighted online tree augmentation problem.}
    \label{alg:rotap}
    \KwIn{A rooted tree $T=(V, E)$, a set of uplinks $L\subseteq\binom{V}{2}$. At each time step $i$, an online request $e_i\in E$ arrives.}
    \KwOut{At each time step $i$, a subset of links $L'_i\subseteq L$ that cover $\{e_1,\dots,e_i\}$ such that $L'_{i-1} \subseteq L'_i$, i.e., we add the set of links $L'_i\setminus L'_{i-1}$ to the existing solution $L'_{i-1}$ so that the request $e_i$ is covered.}
    $L'_0\leftarrow\emptyset$\\
    \For{each arriving request $e_i$ at time step $i$} {
        $L'_i\leftarrow L'_{i-1}$.\\
        \If{$e_i$ is uncovered by $L'_i$} {
        $L'_i\leftarrow L'_{i-1}\cup\{$topmost link covering $e_i\}\cup\{$bottommost link covering $e_i\}$.\\
        }
    }
    \Return{$L'_i$}.
\end{algorithm}

\subsection{Analysis}
Fix an optimal solution $\OPT$. Partition the set of input requests
$\{e_1,\ldots,e_s\}$ into $|\OPT|$ parts, one for each link in $\OPT$, such that the part corresponding to a link $\ell\in\OPT$ contains only requests covered by $\ell$. If a request is covered by multiple links in $\OPT$, assign it arbitrarily to one of their corresponding parts.

Observe that our algorithm selects at most two links for each request. Therefore, to prove~\Cref{thm:online-tap}, it suffices to show that, for every part in the above partition, the algorithm selects at most $O(\log n)$ links. We establish the following lemma.

\begin{lemma}\label{lem:relbranch}
For any set of requests assigned to a link in $\OPT$, \Cref{alg:rotap} selects at most $O(\log n)$ links while serving those requests, where $n=|V|$.
\end{lemma}
\begin{proof} 
    Fix a link $\ell_{\OPT}\in\OPT$, and let $R\subseteq \{e_1,\ldots,e_s\}$ be an arbitrary set of requests covered by $\ell_{\OPT}$. For any link $\ell$, let $P(\ell)$ denote the unique path in the tree $T$ between the endpoints of $\ell$. Since every request in $R$ is covered by $\ell_{\OPT}$, all requests in $R$ lie on the path $P(\ell_{\OPT})$.

Let $\ALG_R$ denote the set of links selected by our algorithm while serving the requests in $R$. We will show that $|\ALG_R|=O(\log n)$.
    
    For each link $b\in\ALG_R$ that is selected as a bottommost link for some uncovered request in $R$, we associate a unique vertex $w(b)$ on the path $P(\ell_{\OPT})$. Specifically, $w(b)$ is the vertex that lies on both $P(\ell_{\OPT})$ and $P(b)$ and is farthest from the root. We refer to $w(b)$ as the \emph{branching vertex} of $b$.

We first show that distinct bottommost links have distinct branching vertices. Let $b_1$ and $b_2$ be two distinct bottommost links selected for requests $r_1$ and $r_2$, respectively, and assume without loss of generality that $r_1$ arrives before $r_2$. Since $b_2$ is selected for $r_2$, the request $r_2$ must be uncovered upon arrival.
Both \emph{edges} $r_1$ and $r_2$ lie on the path $P(\ell_{\OPT})$. We claim that $r_2$ cannot lie above $r_1$ on this path (that is, $r_2$ cannot be closer to the root than $r_1$). Indeed, when $r_1$ arrived, the algorithm selected a topmost link covering $r_1$. Since $\ell_{\OPT}$ also covers $r_1$, it was one of the candidate links, and therefore the selected topmost link must cover every edge of $P(\ell_{\OPT})$ above $r_1$, including $r_2$. Hence $r_2$ would already have been covered before its arrival, a contradiction.
Similarly, $r_2$ cannot lie on the portion of $P(\ell_{\OPT})$ between $r_1$ and the branching vertex $w(b_1)$. By definition, this segment is covered by the bottommost link $b_1$, and thus $r_2$ would again be covered before its arrival. Therefore, $r_2$ must lie strictly below both $r_1$ and $w(b_1)$ on $P(\ell_{\OPT})$. It follows that the branching vertex $w(b_2)$ lies strictly below $w(b_1)$, and hence $w(b_1)\neq w(b_2)$.

\begin{figure}
    \centering
    \includegraphics[scale=1]{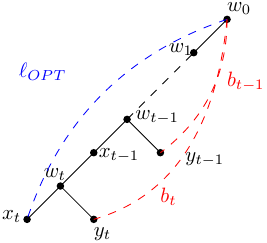}
    \caption{$w_i$ is the branching vertex for the bottommost link $b_i$. $x_i$ and $y_i$ are its two children.}
    \label{fig:simpleTAP}
\end{figure}

Let $w_1,\ldots,w_t$ be the branching vertices induced by the bottommost links selected by the algorithm for requests in $R$, ordered by increasing distance from the root along the path $P(\ell_{\OPT})$. We claim that there are exactly $t$ requests in $R$ that are uncovered upon arrival. To see this, observe that each uncovered request $r \in R$ causes the algorithm to select a distinct bottommost link (if two requests shared the same bottommost link, the one arriving later would already be covered by it, a contradiction). Since we established that distinct bottommost links have distinct branching vertices, the number of uncovered requests in $R$ must equal the number of branching vertices, which is $t$. Since the algorithm selects at most two links for each uncovered request, it follows that

\[
|\ALG_R| \le 2t.
\]

    \paragraph{Bound on $t$.} Let $N(v)$ denote the number of vertices in the subtree rooted at a vertex $v$. We first show that, assuming
\[
N(w_i)\ge 1+2N(w_{i+1})
\]
for every $i\in\{1,\ldots,t-1\}$, we obtain the desired bound on $t$. Indeed, the above inequality immediately implies that
\[
N(w_1)\ge 1+2+2^2+\cdots+2^{t-1}N(w_t)\ge 2^t-1.
\]
Since $N(w_1)\le n$, it follows that
\[
n\ge N(w_1)\ge 2^t-1,
\]
and hence $t=O(\log n)$. Combining this with the bound $|\ALG_R|\le 2t$, we conclude that
\[
|\ALG_R|=O(\log n).
\]

It remains to prove the claim that
\[
N(w_i)\ge 1+2N(w_{i+1})
\]
for every $i\in\{1,\ldots,t-1\}$. Fix any such $i$. We will argue that $w_i$ has two distinguished children. One child lies on the path $P(\ell_{\OPT})$; denote it by $x$.

Let $b$ be the bottommost link whose branching vertex is $w_i$, i.e., $w(b)=w_i$. By the definition of a branching vertex, either $w_i$ is an endpoint of $b$, or the path $P(b)$ branches away from $P(\ell_{\OPT})$ at $w_i$. The first case is impossible, since then $\ell_{\OPT}$ would be a strictly better candidate for the corresponding request, contradicting the choice of $b$ as a bottommost link. Therefore, $P(b)$ must branch away from $P(\ell_{\OPT})$ at $w_i$. Let $y$ denote the child of $w_i$ that lies on $P(b)$.

By the definition of a bottommost link, we have $N(y)\ge N(x)$. Since the subtrees rooted at $x$ and $y$ are disjoint,
\[
N(w_i)\ge 1+N(x)+N(y)\ge 1+2N(x).
\]
Finally, because $w_{i+1}$ lies below $w_i$ on the path $P(\ell_{\OPT})$, the subtree rooted at $w_{i+1}$ is contained in the subtree rooted at $x$. Hence $N(x)\ge N(w_{i+1})$, and therefore
\[
N(w_i)\ge 1+2N(w_{i+1}),
\]
as required.
\end{proof}

\section*{Acknowledgment}
The authors wish to thank Arindam Khan for participating in the initial discussions that helped shape this work.

\bibliographystyle{alpha}
\bibliography{references_tap}

\appendix
\crefalias{section}{appendix}

\section{Reducing Online CAP to Online CacAP}\label{app:captocacap}

In this section, we prove the following proposition.
\CapFromCacap*

To establish the above proposition, we first recall a classical result of
Dinitz, Karzanov, and Lomonosov~\cite{DKL76}, as stated in Theorem~3.1 of
Traub and Zenklusen~\cite{TZ23}.

\begin{theorem}[\cite{DKL76}]\label{thm:cactus-representation}
For any undirected graph $G=(V,E)$, one can efficiently construct a cactus
$\overline{G}=(\overline{V},\overline{E})$ together with a mapping
$\phi: V \to \overline{V}$ such that:
\begin{enumerate}[label=(\roman*)]
\item If $C \subseteq V$ is a minimum cut in $G$, then there exists a $2$-cut
$\overline{C} \subseteq \overline{V}$ in the cactus $\overline{G}$ satisfying
\[
   C = \{ v \in V : \phi(v) \in \overline{C} \}.
\]
\item Conversely, if $\overline{C} \subseteq \overline{V}$ is a $2$-cut in the
cactus $\overline{G}$, then
\[
   C = \{ v \in V : \phi(v) \in \overline{C} \}
\]
is a minimum cut in $G$.
\end{enumerate}
\end{theorem}

To prove the proposition, we assume the existence of a $c$-competitive online
algorithm for cactus augmentation and show how to obtain a $c$-competitive
online algorithm for connectivity augmentation.

Given a (weighted) CAP instance consisting of a $k$-edge-connected graph $G$
(which is not $(k+1)$-edge-connected) and a set of links $\mathcal{L}$, we apply
\Cref{thm:cactus-representation} to efficiently construct a corresponding
(weighted) CacAP instance. In particular, since $G$ is $k$-edge-connected but
not $(k+1)$-edge-connected, every minimum cut in $G$ has size exactly $k$.
Specifically, we obtain a cactus graph $\overline{G}$ together with a mapping
$\phi: V(G) \to V(\overline{G})$, and define the set of links
\[
\overline{\mathcal{L}} := \{ (\phi(u), \phi(v)) \mid (u,v) \in \mathcal{L} \}.
\]
For each online request $(s,t)$ in the CAP instance, we create a corresponding
request $(\phi(s), \phi(t))$ in the CacAP instance.

We run the online algorithm for CacAP on the constructed instance. Whenever the
algorithm selects a link $(\phi(u), \phi(v)) \in \overline{\mathcal{L}}$, our
algorithm for CAP selects the corresponding link $(u,v) \in \mathcal{L}$. Since
costs are preserved, it remains to argue that feasibility is preserved by this
reduction.

\begin{claim}\label{clm:cap-cacap-feasibility}
A set of links $A \subseteq \mathcal{L}$ serves a request $(s,t)$ in $G$ if
and only if the corresponding set of links
$\bar{A} := \{(\phi(u),\phi(v)) \mid (u,v) \in A\}$ serves the request
$(\phi(s), \phi(t))$ in $\overline{G}$.
\end{claim}

\begin{proof}
We prove the claim in two directions.

\paragraph{($\Rightarrow$)}
Assume that $A$ serves the request $(s,t)$ in $G$. We show that $\bar{A}$ serves
the request $(\phi(s),\phi(t))$ in $\overline{G}$. Suppose, for the sake of
contradiction, that $\bar{A}$ does not serve the request
$(\phi(s),\phi(t))$ in $\overline{G}$. Then there exists a $2$-cut
$\overline{M} \subseteq V(\overline{G})$ in the cactus $\overline{G}$ that
separates $\phi(s)$ and $\phi(t)$ and is not covered by any link in $\bar{A}$.

By part~(ii) of \Cref{thm:cactus-representation}, the set
\[
M := \{ v \in V(G) \mid \phi(v) \in \overline{M} \}
\]
is a $k$-cut in $G$. Moreover, since $\overline{M}$ separates $\phi(s)$ and
$\phi(t)$ in $\overline{G}$, the cut $M$ separates $s$ and $t$ in $G$.

Since $A$ serves the request $(s,t)$ in $G$, there is no $k$-cut separating
$s$ and $t$ in $G \cup A$. Hence, there must exist a link $\ell=(u,v) \in A$
that covers the cut $M$. Note that a link $(u,v)$ covers $M$ if and only if the
corresponding link $(\phi(u),\phi(v))$ covers $\overline{M}$. Therefore, the
link $(\phi(u),\phi(v)) \in \bar{A}$ covers the $2$-cut $\overline{M}$,
contradicting the assumption that $\overline{M}$ is not covered by any link in
$\bar{A}$. This contradiction shows that $\bar{A}$ serves the request
$(\phi(s),\phi(t))$ in $\overline{G}$.

\paragraph{($\Leftarrow$)}
Assume that $\bar{A}$ serves the request $(\phi(s),\phi(t))$ in $\overline{G}$.
We show that $A$ serves the request $(s,t)$ in $G$. Suppose, for the sake of
contradiction, that $A$ does not serve the request $(s,t)$ in $G$. Then there
exists a $k$-cut $M \subseteq V(G)$ in $G$ that separates $s$ and $t$ and is
not covered by any link in $A$.

By part~(i) of \Cref{thm:cactus-representation}, there exists a $2$-cut
$\overline{M} \subseteq V(\overline{G})$ in the cactus $\overline{G}$ such that
\[
M = \{ v \in V(G) \mid \phi(v) \in \overline{M} \}.
\]
Since $M$ separates $s$ and $t$ in $G$, the corresponding $2$-cut
$\overline{M}$ separates $\phi(s)$ and $\phi(t)$ in $\overline{G}$. Moreover,
a link $(u,v)$ covers $M$ if and only if the corresponding link
$(\phi(u),\phi(v))$ covers $\overline{M}$. Hence, since no link in $A$ covers
$M$, it follows that no link in $\bar{A}$ covers $\overline{M}$.

This contradicts the assumption that $\bar{A}$ serves
$(\phi(s),\phi(t))$ in $\overline{G}$, since $\overline{M}$ is a $2$-cut
separating $\phi(s)$ and $\phi(t)$ that remains uncovered. Therefore, $A$ must
serve the request $(s,t)$ in $G$.
\end{proof}

\section{Proof of \Cref{prop:transitivity}}\label{sec:transitivity}

In this section, we prove the following proposition.
\transitivity*

\begin{proof}
Fix any pair of vertices $(x,y) \in \mathcal{V}((u,v)) \times \mathcal{V}((u,v))$.
Let $F \subseteq E(G)$ be any $2$-cut in the cactus $G$ that separates $x$ and $y$.
We show that $F$ also separates $u$ and $v$.

By definition of the projection vertex set $\mathcal{V}((u,v))$, every $u$--$v$ path in $G$
contains both $x$ and $y$. Hence, after deleting the edges in $F$, there is no remaining
$x$--$y$ path in $G$, and consequently no $u$--$v$ path remains either. Thus, $F$ also
separates $u$ and $v$.

Since $u$ and $v$ are $3$-edge-connected in the augmented cactus, every $2$-cut separating
$u$ and $v$ is covered by some added link. As every $2$-cut separating $x$ and $y$ also
separates $u$ and $v$, all such $2$-cuts are covered as well. Therefore, $x$ and $y$ are
also $3$-edge-connected.

\end{proof}

\section{Proof of \Cref{prop:CacConCross}} \label{app:connectedCrossing}
We begin by proving the following proposition.

\begin{restatable}{proposition}{connectedCrossing}\label{prop:connectedCrossing}
    Let \(G\) be a cycle and let \(A \subseteq L\) be a set of links added to \(G\).
    A pair of distinct vertices \((u,v)\) is \(3\)-edge-connected in \(G \cup A\) if and only if
    there exists a sequence of crossing links in \(A\) from \(u\) to \(v\).
\end{restatable}

\begin{proof}
We prove the proposition in two directions.

\paragraph{(\(\Leftarrow\))}
    Suppose there exists a sequence of crossing links from \(u\) to \(v\).
    We show that \(u\) and \(v\) are \(3\)-edge-connected in \(G \cup A\).

    Consider any pair of edges \(e_1,e_2\) of the cycle forming a \(2\)-cut that separates
    \(u\) and \(v\). Let \(V_u\) and \(V_v\) denote the vertex sets of the two resulting components,
    containing \(u\) and \(v\), respectively.
    By assumption, there is a sequence of crossing links
    \(\ell_1 - \ell_2 - \cdots - \ell_q\), where one endpoint of \(\ell_1\) is \(u\) and one
    endpoint of \(\ell_q\) is \(v\).

    If none of the links in this sequence crossed the cut \(\{e_1,e_2\}\), then there would exist
    an index \(i\) such that the links \(\ell_1,\ell_2,\dots,\ell_i\) have both endpoints in \(V_u\),
    while \(\ell_{i+1}\) has both endpoints in \(V_v\).
    However, in this case \(\ell_i\) and \(\ell_{i+1}\) cannot cross, contradicting the assumption
    that consecutive links in the sequence are crossing.

    Therefore, every \(2\)-cut separating \(u\) and \(v\) in the cycle is covered by at least one
    link in \(A\). This implies that no \(2\)-cut separates \(u\) and \(v\) in \(G \cup A\), and hence
    \(u\) and \(v\) are \(3\)-edge-connected.

    \paragraph{(\(\Rightarrow\))}
Assume that \(u\) and \(v\) are \(3\)-edge-connected in \(G \cup A\).
We show that there exists a sequence of crossing links in \(A\) from \(u\) to \(v\).

Label the vertices of the cycle by their clockwise distance from \(u\).
Initially, assign the color blue to \(u\).
While there exists an (as yet unselected) link with at least one endpoint colored blue,
assign the color blue to all vertices whose labels lie between the labels of the two endpoints
of that link (both inclusive).
Next, assign the color red to \(v\) and repeat the same coloring procedure for the red color.

Let \(B\) and \(R\) denote the sets of vertices colored blue and red, respectively.
Note that a vertex (including \(u\) or \(v\)) may receive both colors, and that the two endpoints
of every link are always assigned the same set of colors (possibly none, one color, or both).

Observe that the set \(B\) forms a contiguous clockwise path starting at \(u\) and ending at some
vertex \(u'\), and similarly \(R\) forms a contiguous path containing \(v\).
Let \(x\) be the vertex with the smallest label that is colored red.
If \(x\) is ever assigned the color blue during the coloring process, then all vertices that
are subsequently colored red will also be colored blue.
Hence, if \(R\cap B\neq\emptyset\), then \(R\subseteq B\).

We now analyze the possible interactions between \(R\) and \(B\); see \Cref{fig:connCross}
for an illustration.

\begin{figure}
\centering
\begin{subfigure}{0.3\textwidth}
    \centering
    \includegraphics[scale=0.8,page=1]{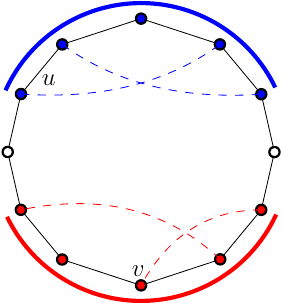}
    \caption{No intersection.}
    \label{fig:cc_a}
\end{subfigure}
\begin{subfigure}{0.3\textwidth}
    \centering
    \includegraphics[scale=0.78,page=2]{fig_crossing_connect.pdf}
    \caption{$R=B$.}
    \label{fig:cc_b}
\end{subfigure}
\begin{subfigure}{0.3\textwidth}
    \centering
    \includegraphics[scale=0.72,page=3]{fig_crossing_connect.pdf}
    \caption{One contains the other.}
    \label{fig:cc_c}
\end{subfigure}
\caption{Different ways in which sections of colored vertices can interact.}
\label{fig:connCross}
\end{figure}

\begin{itemize}
\item \emph{\(R\cap B=\emptyset\)} (\Cref{fig:cc_a}).
Consider the cut \((B, V(G)\setminus B)\), which separates \(u\) and \(v\).
By construction, every link with one endpoint in \(B\) has its other endpoint also in \(B\),
and the same holds for \(R\).
Thus, no link in \(A\) covers this cut, implying that the cut has size two in \(G \cup A\),
contradicting the assumption that \(u\) and \(v\) are \(3\)-edge-connected.

\item \emph{\(B=R\)} (\Cref{fig:cc_b}).
We claim that there exists a sequence of crossing links from \(u\) to \(v\).
Let \(w\) be the colored vertex with the largest label.
By the coloring procedure, there exists a sequence of links starting at \(u\) and ending at \(w\),
and similarly a sequence of links starting at \(v\) and ending at \(w\).
The last links in these two sequences both have \(w\) as an endpoint and therefore cross.
Concatenating these two sequences yields a sequence of crossing links from \(u\) to \(v\).

\item \emph{\(R\subsetneq B\)} (\Cref{fig:cc_c}).
Consider the cut \((R, V(G)\setminus R)\).
Vertex \(v\) is colored both red and blue, while vertex \(u\) is colored only blue.
(Indeed, if \(u\) were ever colored red, the coloring process would imply \(B\subseteq R\),
contradicting the strict containment.)
Hence, this cut separates \(u\) and \(v\).
Since the endpoints of every link receive the same set of colors, no link covers the cut
\((R, V\setminus R)\).
Thus, this cut has size two in \(G \cup A\), again contradicting that \(u\) and \(v\)
are \(3\)-edge-connected.
\end{itemize}

The above contradictions leave only the case \(B=R\), which guarantees the existence of
a sequence of crossing links from \(u\) to \(v\).

\end{proof}

We are now ready to prove the following proposition.

\CanConCross*
\begin{proof}
If \(u\) and \(v\) lie in the same cycle, the claim follows directly from
\Cref{prop:connectedCrossing}. Hence, we assume that \(u\) and \(v\) lie in different
cycles. We prove the proposition in two directions.

\paragraph{(\(\Rightarrow\))}
Assume that \(u\) and \(v\) are \(3\)-edge-connected.
We show that there exists a sequence of crossing links from \(u\) to \(v\).

Let \(v_1=u, v_2, \ldots, v_q=v\) be the vertices that lie on every \(u\)--\(v\) path;
note that \(v_2,\ldots,v_{q-1}\) are articulation vertices.
Each consecutive pair \((v_i,v_{i+1})\) lies on a common cycle and admits three
edge-disjoint paths between them by \Cref{prop:transitivity}.

Suppose, for contradiction, that no sequence of crossing links exists from \(u\) to \(v\).
Then there must exist an index \(i\) such that no sequence of crossing links exists
from \(v_i\) to \(v_{i+1}\); otherwise, concatenating such sequences for all consecutive
pairs would yield a sequence from \(u\) to \(v\).

Fix such a pair \((v_i,v_{i+1})\), and let \(C\) be the cycle containing them.
Consider the projections of all added links onto the cycle \(C\).
Since \(v_i\) and \(v_{i+1}\) admit three edge-disjoint paths in \(G \cup A\),
they are \(3\)-edge-connected in the graph induced by \(C\) together with these
projected links.
By \Cref{prop:connectedCrossing}, this implies the existence of a sequence of crossing
links from \(v_i\) to \(v_{i+1}\), contradicting our assumption.
Therefore, a sequence of crossing links from \(u\) to \(v\) must exist.

\paragraph{(\(\Leftarrow\))}
Assume that there exists a sequence of crossing links
\(\ell_1-\ell_2-\cdots-\ell_t\) from \(u\) to \(v\), where \(u\) lies in the projection
vertex set of \(\ell_1\) and \(v\) lies in the projection vertex set of \(\ell_t\).
We show that \(u\) and \(v\) are \(3\)-edge-connected.

Let \(e_1,e_2\) be any two edges of the cactus forming a \(2\)-cut that separates \(u\)
and \(v\), and let \(V_u\) and \(V_v\) be the vertex sets of the two connected components
of \(G-\{e_1,e_2\}\) containing \(u\) and \(v\), respectively.
Note that \(e_1\) and \(e_2\) lie on a common cycle.
We claim that some link \(\ell_i\) in the sequence covers this \(2\)-cut.

Suppose not.
Then \(\ell_1\) does not cover the cut, and hence its endpoints lie on the same side.
They cannot both lie in \(V_v\), since in that case they would be connected within the
component \(V_v\) of \(G-\{e_1,e_2\}\), yielding an endpoint-to-endpoint path avoiding
\(u \in V_u\), contradicting that \(u\) lies in the projection vertex set of \(\ell_1\).
Therefore, both endpoints of \(\ell_1\) lie in \(V_u\).
By the same argument, both endpoints of \(\ell_t\) lie in \(V_v\).

Since no link in the sequence covers the cut, there exists an index \(i\) such that
\(\ell_1,\ldots,\ell_i\) have both endpoints in \(V_u\), while \(\ell_{i+1}\) has both
endpoints in \(V_v\).
But then the cut \((V_u,V_v)\) separates the endpoints of \(\ell_i\) and \(\ell_{i+1}\),
so these two links cannot cross, contradicting the assumption that consecutive links in
the sequence are crossing.

Thus, every \(2\)-cut separating \(u\) and \(v\) is covered by at least one link, implying
that \(u\) and \(v\) are \(3\)-edge-connected.
\end{proof}

\end{document}